\documentclass[12pt,a4paper]{amsart}

\usepackage{amssymb,hyperref,mathrsfs,upref}
\usepackage[margin=1in]{geometry}

\theoremstyle{plain}
\newtheorem{prop}{Proposition}

\theoremstyle{definition}

\everymath{\displaystyle}
\allowdisplaybreaks
\numberwithin{equation}{section}

\newcommand{\R}{\mathbb{R}}
\newcommand{\wh}[1]{\widehat{#1}}

\DeclareMathOperator{\erf}{erf}
\DeclareMathOperator{\arctanh}{arctanh}

\newcommand{\UniBA}{%
   Universit\`{a} degli Studi di Bari ``Aldo Moro'',
   Department of~Economics and Finance,
   Largo Abbazia S.~Scolastica,
   Bari, I-70124 Italy}

\begin{document}

\title[Challenges in approximating the~Black and~Scholes call formula]
   {Challenges in approximating\\the~Black and~Scholes call formula\\with hyperbolic tangents}

\author{Michele Mininni}
\address{\UniBA}
\email{michele.mininni@uniba.it}

\author{Giuseppe Orlando}
\address{\UniBA}
\curraddr{Universit\`{a} degli Studi di Camerino, School of~Science and Technologies,
   Via M.~delle Carceri~9, Camerino, I-62032 Italy}
\email{giuseppe.orlando@uniba.it, giuseppe.orlando@unicam.it}

\author{Giovanni Taglialatela}
\address{\UniBA}
\email{giovanni.taglialatela@uniba.it}

\begin{abstract}
	In this paper we introduce the concept of standardized call function and we obtain a new approximating formula for the Black and~Scholes call function through the hyperbolic tangent. This formula is~useful for pricing and risk management as well as for extracting the~implied volatility from quoted options. The latter is~of~particular importance since it indicates the~risk of~the underlying and it is~the main component of~the~option's price. Further we~estimate numerically the approximating error of the~suggested solution and, by comparing our results in~computing the~implied volatility with the~most common methods available in~literature we discuss the challenges of this approach.
\end{abstract}


\maketitle

\sloppy

\section*{Preface}

This is a reduced version of the paper without tables and figures.

{}For investors and traders a key component in their decision making is to assess the risk they run. A common way to do so is to recur to the dispersion of returns. However this measure has the problem that is computed on past performances and may have little to do with the current level of risk.
{}Within the Black-Scholes (BS) framework~\cite{Black1973}, later extended by Merton~\cite{Merton1973}, it is possible to identify a relation between the value of an asset and the option written on it. This is expressed through a link between the~option's price and some factors such as time, rates, dividends and, above all, volatility~\cite{Hull2006}. Since then, as the usage of the BS formula has became widespread in financial markets, options are traded and priced in terms of their risk or, in other terms, of the so-called implied~volatility (i.e.~the actual volatility embedded in the option's price).

For the reasons above mentioned a key issue has become the analytical tractability of the said formula for pricing, hedging and inversion (needed to find the BS implied volatility). 

{}On the other hand, given the structure of the BS formula that cannot be analytically inverted, the implied volatility can be found only through numerical approximation methods. However, in some instances, even those methods may fail for technical reasons~\cite{Orlando2017}.

In this work we~show a closed form formula for approximating the~BS call formula by means of~the hyperbolic tangent. This allows us to~determine both the value of~the call for any change of~the key variables and to~derive the~implied volatility at once for all possible combinations of~underlying, strike, time, etc.
To~achieve the~aforementioned result, we~introduce the~so called ``standardized call function'', which is a single-parameter function representing the~general family of~calls.

{}The paper is~organized as follows. The first section provides some background on both the~topic and the~literature and defines the~notations used in~the rest of~the paper.
The second and third sections  propose some approximations of  the call function in both  the cases $S \ne X$  and $ S = X$. The fourth section is devoted to derive the approximations of the implied volatility again in the cases $S \ne X$  and $ S = X$. 
In the fifth section we present some  results of numerical simulations  and compare such results with the ones given by other authors. 
Finally the last section summarizes this work and draws the~line for future research.

\section{Background}

\subsection{Literature review} \label{SS:LitReview}

Volatility on the markets has been calculated by many. Mo \& Wu (2007)~\cite{Mo2007},
on a sample from January~3, 1996 to~August 14, 2002, 346 weekly observations for each index,
reported that the~implied volatility on the~S\&P 500 Index (SPX), the~FTSE 100 Index (FTS),
and the~Nikkei-225 Stock Average (NKY) ranges from 15\% to~35\% with $S/X$ comprised between~80\% and~120\% and maturity 1m, 3m, 6m and 12m.
Glasserman \&~Wu (2011)~\cite{Glasserman2011}, on a sample consisting of~2049 data points
from August~9, 2001 to~June~16, 2009, found that the~implied volatility ranges from 5\% to~43\% with currency option on EURUSD, GBPUSD and USDJPY for at 10-delta put (P10d), 25-delta put (P25d), At-The-Money (ATM), 25-delta call (C25d), and 10-delta call (C10d).

Finally Figure \ref{fig:VIX} below shows the~overall daily distribution of~the VIX~\cite{CBOE1} since inception up to~2015. Throughout these years median volatility has been $\sim 18\%$, mode $\sim 13\%$ and average is~$\sim 20\%$ and high values are rare. 

\begin{figure} \label{fig:VIX}
\caption{}
\end{figure}

{}For computing both the~call price as well as the~implied volatility
several numerical algorithms are available~\cite{Dura2010}, \cite{Orlando2017}.
However there is~a cost and some drawbacks for those techniques which motivate, where possible,
the search for closed form approximations (see for example Manaster and Koehler (1982)~\cite{Manaster1982}).
Notwithstanding in~using those approximations one must be careful.
For example Estrella (1995)~\cite{Estrella1995} showed that, for a plausible range of~parameter values,
the Taylor series for the~BS formula diverges and ``even when the~series converges,
finite approximations of~very large order are generally necessary to~achieve acceptable levels of~accuracy''.
The alternative of~using the~exact formula with predetermined changes of~the underlying,
while it provides more accurate results, presents some drawbacks in~terms of~computations and loss of~flexibility~\cite{Estrella1995}.
Therefore ``simple solutions are desirous because they have two very attractive properties.
They are easy to~implement, and provide very fast computational algorithm''~\cite{Hofstetter2001}.

{}Among the~closed form approximations that rely on Taylor approximations or on the~power series expansion
of the~cumulative normal distribution function (cndf), there are Brenner \&~Subrahmanyam (1988)~\cite{Brenner1998},
Bharadia, Christofides \&~Salkin (1995)~\cite{Bharadia1995}, Chance (1996)~\cite{Chance1996},
Corrado \&~Miller (1996)~\cite{Corrado1996}~\cite{Corrado1996a}, Liang \&~Tahara (2009)~\cite{Liang2009}, Li (2005)~\cite{Li2005}
(for a review on the~subject see Orlando \&~Taglialatela (2017)~\cite{Orlando2017}).  

{}Other closed-form approximations that work only in~the At-The-Money case are the~P\'{o}lya approximation~\cite{Polya1949},~\cite{Matic2017},
and the~logistic approximation~\cite{Pianca2005}.
For~these, it has been shown that the~first is~remarkably accurate for a very large range of~parameters,
while the~second one is~accurate only for extreme maturities~\cite{Pianca2005}. 

{}Finally, Hofstetter \&~Selby (2001)~\cite{Hofstetter2001} obtained approximations
by replacing the~cndf by the~logistic distribution and Li (2008)~\cite{Li2008} developed a closed-form method based on the~rational functions.

\subsection{The Black and~Scholes Formula}

The BS formula for deriving the~price $C$ of a~European call option
is~described by   
\begin{equation} \label{E-BS-formula}
C :=  S\,N(d_1) - X\,N(d_2) \, , 
\end{equation}
where
\begin{itemize}
\item  $S$ is~the value of~the underlying,
\item  $X = K\,e^{-rT}$ is~the present value of~the strike price,
\item  $K$ is~the strike price,
\item  $r$ is~the interest rate,
\item  $T$ is~the time to~maturity in~terms of~a year,
\item  $N(x)$ is~the cumulative distribution function of~the standard normal i.e.
\[
N(x)
:= \frac{1}{\sqrt{2\pi\,}} \int_{-\infty}^x e^{-t^2/2} \, dt
\]

\item  $d_1 := \frac{\log(S/X)}{\sigma\sqrt{T\,}} + \frac{\sigma}{2}\sqrt{T\,}$
	is the~first parameter of~probability i.e.~``the factor by which the~present value of~contingent receipt of~the stock,
	contingent on exercise, exceeds the~current value of~the stock''~\cite{Nielsen1992},

\item  $d_2 := \frac{\log(S/X)}{\sigma\sqrt{T\,}} - \frac{\sigma}{2}\sqrt{T\,}$
	is the~second parameter of~probability which represents the~risk-adjusted probability of~exercise,

\item  $\sigma$ is~the volatility.
\end{itemize}

It is~worth noting that, given the~parameters $S, X, T$,
(or, equivalently $S, K, r, T$)
the price $C$ of the call  is a function $C=\mathscr{C}(\sigma)$ of the volatility $\sigma$;
the~implied volatility is~obtained by inverting such a function. In the following sections we shall obtain a suitable approximation  $\wh{\mathscr{C}}$ of $\mathscr{C}$ for all $S, X, T$. The implied volatility will be approximated by inverting $\wh{\mathscr{C}}$.

\section{Approximating the~call function when $S \ne X$}

\subsection{The standardized call function}

In order to~simplify the~presentation
we introduce a~family of standardized  call functions: 
\[
\chi_\alpha(x)
  :=  N\biggl(\frac{\alpha}2 \Bigl(x-\frac{1}{x}\Bigr) \biggr)
      -  e^{\alpha^2/2} \, N\biggl(-\frac{\alpha} 2\Bigl(x+\frac{1}{x}\Bigr)\biggr) \, ,
\qquad
x>0 \, , 
\]
depending on a~single parameter~$\alpha>0$.

The following Proposition contains the main properties of the mappings $\chi_\alpha$.

\begin{prop} \label{P-1}
For all $\alpha>0$ one has that:
\begin{enumerate}
\item[\textup{(i)}]    $\lim_{x\to0^+} \chi_\alpha(x) = 0$ and $\lim_{x\to+\infty} \chi_\alpha(x) = 1$.
\item[\textup{(ii)}]   $\chi_\alpha(x)$ is~strictly increasing in~$\left]0,+\infty\right[$.
\item[\textup{(iii)}]  $\chi_\alpha(x)$ is~strictly convex in~$\left]0,1\right]$ and~strictly concave in~$\left[1,+\infty\right[$.
\end{enumerate}
\end{prop}

\begin{proof}[Proof of~\textup{(i)}]
It's a trivial consequence of~the limits at $\pm\infty$ of~$N(x)$.
\end{proof}

\begin{proof}[Proof of~\textup{(ii)}]
We have:
\begin{align*}
\chi_\alpha'(x)
&  =  \frac{\alpha}{2\,\sqrt{2\pi\,}} \,
       \exp\biggl[-\frac{\alpha^2}{8}\Bigl(x-\frac{1}{x}\Bigr)^2\biggr]
       \Bigl(1 + \frac{1}{x^2}\Bigr)  \\
&  \qquad
   +  \frac{\alpha \, e^{\alpha^2/2}}{2\,\sqrt{2\pi\,}} \,
      \exp\biggl[-\frac{\alpha^2}{8}\Bigl(x+\frac{1}{x}\Bigr)^2\biggr]
      \Bigl(1 - \frac{1}{x^2}\Bigr)
\end{align*}
and, since
\[
\exp\biggl[-\frac{\alpha^2}{8}\Bigl(x-\frac{1}{x}\Bigr)^2\biggr]
=  e^{\alpha^2/2} \, \exp\biggl[-\frac{\alpha^2}{8}\Bigl(x+\frac{1}{x}\Bigr)^2\biggr]
\]
we get
\begin{equation} \label{E-fp}
\chi_\alpha'(x)
  =  \frac{\alpha}{\sqrt{2\,\pi\,}} \, \exp\biggl[-\frac{\alpha^2}{8}\Bigl(x-\frac{1}{x}\Bigr)^2\biggr] \, ,
\end{equation}
from which we~derive that $\chi_\alpha(x)$ is~strictly increasing in~$\left]0,+\infty\right[$.
\end{proof}

\begin{proof}[Proof of~\textup{(iii)}]
Differentiating \eqref{E-fp} we~get
\begin{equation} \label{E-fpp}
\chi_\alpha''(x)
=  \frac{\alpha^3}{4\,\sqrt{2\,\pi\,}} \,
\exp\biggl[-\frac{\alpha^2}{8}\Bigl(x-\frac{1}{x}\Bigr)^2\biggr] \frac{1-x^4}{x^3} \, ,
\end{equation}
from which we~derive that $\chi_\alpha(x)$ is~strictly convex in~$\left]0,1\right]$ and~strictly concave in~$\left[1,+\infty\right[$.
\end{proof}

For $S$, $X$ and $T$ fixed,
the relationship between the call function $C= \mathscr{C}(\sigma) $ 
and the family of functions $(\chi_\alpha)_{\alpha>0}$ is contained in the following

\begin{prop} \label{P-2}
Let us fix $S>0$, $X>0$ and $T>0$, with $X \ne S$,
and let us put  
\[
\alpha
  :=  \sqrt{2\,\bigl|\log(S/X)\bigr|\,} \, ;
\]
then we~have
\[
\mathscr{C}(\sigma)
  =  \begin{cases}
     S\,\chi_\alpha\Bigl(\frac{\sigma\,\sqrt{T\,}}{\alpha}\Bigr)
     &  \text{if $X>S$} \, ,  \\*[2mm]
     S-X + X\,\chi_\alpha\Bigl(\frac{\sigma\,\sqrt{T\,}}{\alpha}\Bigr)
     &  \text{if $X<S$} \, .
\end{cases}
\]
\end{prop}

\begin{proof}
Assume $X>S$;
since ${\alpha^2}/{2} = \log(X/S) = -\log(S/X)$, and therefore $X=Se^{\alpha^2/2}$,  we~have
\begin{align*}
S \, \chi_\alpha\Bigl(\frac{\sigma\,\sqrt{T\,}}{\alpha}\Bigr)
&  =  S \, N\Bigl(-\frac{\alpha^2}{2\,\sigma\,\sqrt{T\,}}
+ \frac{\alpha}{2}\,\frac{\sigma\,\sqrt{T\,}}{\alpha}\Bigr)
-  S\,e^{\alpha^2/2} \,
N\Bigl(-\frac{\alpha^2}{2\,\sigma\,\sqrt{T\,}}
- \frac{\alpha}{2}\,\frac{\sigma\,\sqrt{T\,}}{\alpha}\Bigr)  \\
&  =  S \, N\Bigl(\frac{\log(S/X)}{\sigma\,\sqrt{T\,}} + \frac{\sigma\,\sqrt{T\,}}{2}\Bigr)
-  X \, N\Bigl(\frac{\log(S/X)}{\sigma\,\sqrt{T\,}} - \frac{\sigma\,\sqrt{T\,}}{2}\Bigr)   =  \mathscr{C}(\sigma) \, .
\end{align*}

Assume $X<S$;
since $\alpha^2/2 = \log(S/X) = -\log(X/S)$ and therefore $S= X e^{\alpha^2/2}$,
we~have
\begin{align*}
X \, \chi_\alpha\Bigl(\frac{\sigma\,\sqrt{T\,}}{\alpha}\Bigr)
&  =  X \, N\Bigl(-\frac{\alpha^2}{2\,\sigma\,\sqrt{T\,}} + \frac{\alpha}{2}\,\frac{\sigma\,\sqrt{T\,}}{\alpha}\Bigr)
      -  X\,e^{\alpha^2/2} \,
N\Bigl(-\frac{\alpha^2}{2\,\sigma\,\sqrt{T\,}} - \frac{\alpha}{2}\,\frac{\sigma\,\sqrt{T\,}}{\alpha}\Bigr)  \\
&  =  X \, N\Bigl(-\frac{\log(S/X)}{\sigma\,\sqrt{T\,}} + \frac{\sigma\,\sqrt{T\,}}{2}\Bigr)
-  S \, N\Bigl(- \frac{\log(S/X)}{\sigma\,\sqrt{T\,}} - \frac{\sigma\,\sqrt{T\,}}{2}\Bigr) \, .
\end{align*}
Now, since
\[
N(x) = 1 - N(-x) \, ,
\qquad
\text{for any $x\in\R$} \, ,
\]
we~get
\begin{align*}
X \, \chi_\alpha\Bigl(\frac{\sigma\,\sqrt{T\,}}{\alpha}\Bigr)
&  =  X - X \, N\Bigl(\frac{\log(S/X)}{\sigma\,\sqrt{T\,}} - \frac{\sigma\,\sqrt{T\,}}{2}\Bigr)
-  S  + S\, N\Bigl(\frac{\log(S/X)}{\sigma\,\sqrt{T\,}} + \frac{\sigma\,\sqrt{T\,}}{2}\Bigr)  \\
&  =  X - S + \mathscr{C}(\sigma) \, . \qedhere
\end{align*}
\end{proof}

\subsection{Construction of~the approximating functions}

{}From~Proposition~\ref{P-2} it follows that,
in order to get a good approximation of the call functions $\mathscr{C}(\sigma)$,
it is sufficient to give, for all $\alpha>0$, a good approximation $\wh{\chi}_\alpha$ of~$\chi_\alpha$.

For the~sake of~simplicity, let us fix $\alpha>0$ so that we can omit the index $\alpha$
and simply denote $\chi$ and~$\wh{\chi}$ the mappings $\chi_\alpha$ and $\wh{\chi}_\alpha$.

By Proposition~\ref{P-1},
we~know that $\chi$ has a~sigmoidal shape.
This fact suggests us to look for an approximation based on the hyperbolic tangent 
\begin{equation} \label{E-tanh}
\tanh(x)
  =  \frac{e^x-e^{-x}}{e^x+e^{-x}}
  =  \frac{e^{2x}-1}{e^{2x}+1} 
\end{equation}
which has a similar shape and has the advantage of having a very simple inverse function 
\begin{equation} \label{E-arctanh}
\arctanh(x)
  =  \frac{1}{2} \, \log\Bigl(\frac{1+x}{1-x}\Bigr) \, .
\end{equation}

To be precise we look for an approximating function of the form
\begin{equation} \label{E-fcirc}
\wh{\chi}(x)
  :=  \frac{1}{2} + \frac{1}{2}\,\tanh\bigl(\varphi(x)\bigr)
   =  \frac{e^{2\,\varphi(x)}}{e^{2\,\varphi(x)}+1}
\end{equation}
where $\varphi \colon \left]0,+\infty \right[ \to \R$ is strictly increasing
and satisfies the conditions $\varphi(0^+) = -\infty$ and $\varphi(+\infty) = + \infty$,
so that $\wh{\chi}$ is strictly increasing and tends to $0$ as $x$ tends to $0$
and tends to $1$ as $ x$ tends to $+\infty$.

For example we can choose $\varphi$ of the form 
\begin{equation}\label{E-phi}
\varphi(x) := c_1\,x-\frac{c_2}{x}+c_3
\end{equation}
with $c_1>0$ and $ c_2>0$ so that
$\varphi(x)$ is strictly increasing and has the desired behaviour at $-\infty$ and $+\infty$.

Obviously, we have to choose the constants $c_1, c_2, c_3$ in such a way that $\wh{\chi}$
gives the best approximation of $\chi$;
hence we impose that both $\wh{\chi}$ and $\chi$ have an inflection point at $x=1$
with the same tangent lines there, i.e.\ we impose that $c_1$, $c_2$ and $c_3$ have to satisfy the conditions: 
\begin{align}
\wh{\chi}(1)
&  =  \chi(1) \, , \label{E-cond0}  \\
\wh{\chi}'(1)
&  =  \chi'(1) \, , \label{E-cond1}  \\
\wh{\chi}''(1)
&  =  \chi ''(1)
   =  0 \, . \label{E-cond2}
\end{align}

As
\[
\tanh'(x) = 1-\tanh^2(x) \, ,
\]
we have
\begin{align*}
\wh{\chi}'(x)
&  =  \frac{1}{2} \, \Bigl[1-\tanh^2\bigl(\varphi(x)\bigr)\Bigr]\,\varphi'(x) \, ,  \\
\wh{\chi}''(x)
&  =  -  \tanh\bigl(\varphi(x)\bigr)\Bigl[1-\tanh^2\bigl(\varphi(x)\bigr)\Bigr]\,[\varphi'(x)]^2
+  \frac{1}{2} \, \Bigl[1-\tanh^2\bigl(\varphi(x)\bigr)\Bigr]\,\varphi''(x)  = \\
&  =  \frac{1}{2} \, \Bigl[1-\tanh^2\bigl(\varphi(x)\bigr)\Bigr] \,
      \biggl[\varphi''(x) - 2\,\tanh\bigl(\varphi(x)\bigr)\,\bigl[\varphi'(x)\bigr]^2\biggr] \, .
\end{align*}
Thus conditions~\eqref{E-cond0}-\eqref{E-cond2} give
\[
\begin{cases}
\frac{1}{2} + \frac{1}{2}\,\tanh(c_1-c_2+c_3) = \chi(1)  \\*[3mm]
\frac{1}{2} \, \Bigl[1-\tanh^2(c_1-c_2+c_3)\Bigr]\,(c_1+c_2) = \chi'(1)  \\*[3mm]
-2\,c_2 - 2\,\tanh(c_1-c_2+c_3)\,(c_1+c_2)^2 = 0 \, .
\end{cases}
\]

The first equation gives
\[
c_1-c_2+c_3 = \arctanh\bigl(2\,\chi(1)-1\bigr)
\]
and we~can rewrite the~second and third equations as
\begin{gather}
\frac{1}{2} \, \Bigl[1-\bigl(2\,\chi(1)-1\bigr)^2\Bigr]\,(c_1+c_2) = \chi'(1) \, ,  \label{E-43} \\*[3mm]
-2\,c_2 - 2\,\bigl(2\,\chi(1)-1\bigr)\,(c_1+c_2)^2 = 0 \, . \label{E-44}
\end{gather}

{}From \eqref{E-43} we~get
\[
c_1+c_2
  =  \frac{\chi'(1)}{2\,\chi(1)\,\bigl(1-\chi(1)\bigr)} \, ;
\]
thus~\eqref{E-44} gives
\begin{equation} \label{E-c2}
c_2 = \frac{\bigl(1 - 2\,\chi(1)\bigr)\,\bigl[\chi'(1)\bigr]^2}{4\,\chi^2(1)\,\bigl(1-\chi(1)\bigr)^2}
\end{equation}
and consequently
\begin{align}
c_1
&  =  \frac{\chi'(1)\,\bigl[2\,\chi(1)\,\bigl(1-\chi(1)\bigr) - \bigl(1 - 2\,\chi(1)\bigr)\,\chi'(1)\bigr]} {4\,\chi^2(1)\,\bigl(1-\chi(1)\bigr)^2} \, ,  \label{E-c1}  \\
c_3
&  =  \arctanh\bigl(2\,\chi(1)-1\bigr)
      +  \frac{\chi'(1)\,\bigl[\bigl(1 - 2\,\chi(1)\bigr)\,\chi'(1) - \chi(1)\,\bigl(1-\chi(1)\bigr)\bigr]}
              {2\,\chi^2(1)\,\bigl(1-\chi(1)\bigr)^2} \, . \label{E-c3}
\end{align}

Hence $c_1$, $c_2$ and $c_3$ are uniquely determined and depend only on $\chi(1)$ and $\chi'(1)$;
{}from this it follows that $c_1$, $c_2$ and $c_3$ depend only on $\alpha$, since
\begin{align}
\chi(1)
&  =	N(0) - e^{\alpha^2/2} \, N(-\alpha)
=  \frac{1}{2} - e^{\alpha^2/2} \, N(-\alpha) \, , \label{E-f1}  \\
\chi'(1)
&  =  \frac{\alpha}{\sqrt{2\,\pi\,}} \, . \label{E-fp1}
\end{align}

We have to~check that $c_1$ and $c_2$ are positive.

We~begin with $c_2$.
According to~\eqref{E-f1} we~have
\begin{equation} \label{E-T0}
1 - 2\,\chi(1) = 2\,e^{\alpha^2/2} \, N(-\alpha) > 0\, ; 
\end{equation}
from this and from \eqref{E-c2} it follows $c_2>0$.

The proof of~the positivity of~$c_1$ is~a little more involved.
First of~all we~remark that
\begin{align*}
2\,\chi(1)  &  \,\bigl(1-\chi(1)\bigr) - \bigl(1 - 2\,\chi(1)\bigr)\,\chi'(1)  \\
&  =  2\,\Bigl(\frac{1}{2} - e^{\alpha^2/2} \, N(-\alpha)\Bigr)\,\Bigl(\frac{1}{2}+e^{\alpha^2/2} \, N(-\alpha)\Bigr)
        -  2\,e^{\alpha^2/2} \, N(-\alpha)\,\frac{\alpha}{\sqrt{2\,\pi\,}}  \\
&  =  \frac{1}{2} - 2\,e^{\alpha^2} \, N^2(-\alpha)
        -  2\,e^{\alpha^2/2} \, N(-\alpha)\,\frac{\alpha}{\sqrt{2\,\pi\,}}  \\
&  =  \frac{1}{2}\Bigl(1+\frac{\alpha^2}{2\,\pi}\Bigr) 
      - \frac{1}{2}\Bigl(\frac{\alpha}{\sqrt{2\,\pi\,}} + 2\,e^{\alpha^2/2} \, N(-\alpha) \Bigr)^2  \\
&  =  \frac{1}{2}
      \Biggl[\sqrt{1+\frac{\alpha^2}{2\,\pi}\,} + \frac{\alpha}{\sqrt{2\,\pi\,}} + 2\,e^{\alpha^2/2} \, N(-\alpha) \Biggr]
      \Biggl[\sqrt{1+\frac{\alpha^2}{2\,\pi}\,} - \frac{\alpha}{\sqrt{2\,\pi\,}} - 2\,e^{\alpha^2/2} \, N(-\alpha) \Biggr] \, .
\end{align*}

Since the~first factor is~positive,
in~order to~prove that $c_1>0$,
we need to~show that the~second factor is~positive too,
that is:
\begin{equation} \label{E-KP}
2\,e^{\alpha^2/2} \, N(-\alpha)
  \le  \sqrt{\frac{\alpha^2}{2\,\pi} + 1\,} - \frac{\alpha}{\sqrt{2\,\pi\,}} \, ,
\qquad
\text{for any $\alpha>0$} \, .
\end{equation}

Inequality~\eqref{E-KP} is~an immediate~consequence of~the Komatsu-Pollak estimate~\cite{Komatsu1955}~\cite{Pollak1956}.

\subsection{The approximating Call functions}

Hence, for all $\alpha>0$ we have found a good approximation to~$\chi_\alpha$,
namely the function 
\begin{equation} \label{E-Chatalpha}
\wh{\chi}_\alpha (x)
  =  \frac{1}{2} \Bigl( 1 + \tanh\bigl (\varphi_\alpha (x)\bigr) \, \Bigr)
\qquad
\text{with}
\qquad
\varphi_\alpha(x)
  :=  c_1(\alpha) \, x - \frac{c_2(\alpha)}{x} + c_3(\alpha) \, ,
\end{equation}
where $c_1(\alpha)$, $c_2(\alpha)$ and $c_3(\alpha)$ are given by 
\eqref{E-c1}, \eqref{E-c2} and \eqref{E-c3} with $\chi$ replaced by $\chi_\alpha $. 

{}From this and from~Proposition~\ref{P-2},
we~can conclude that for all $S>0$ and $X>0$, with $S\ne X$,
a~good approximation to~$\mathscr{C}(\sigma)$ is the function  
\begin{equation} \label{E-Chat}
\widehat{\mathscr{C}}(\sigma)
  :=  \begin{cases}
      S \, \wh{\chi}_\alpha\Bigl(\frac{\sigma\,\sqrt{T\,}}{\alpha}\Bigr) \qquad 
      &  \text{if $X>S$}  \\*[2mm]
      S-X + X\,\wh{\chi}_\alpha\Bigl(\frac{\sigma\,\sqrt{T\,}}{\alpha}\Bigr) \qquad \qquad 
      &  \text{if $X<S$}
\end{cases}
\end{equation}
with
\[
\alpha
  :=  \sqrt{2\,\bigl|\log(S/X)\bigr|\,} \, .
\]

\section{Approximating the~call function when $S=X$} \label{S:S=X}

In the~special case $S=X$,
we~have $d_1(\sigma) = \frac{\sigma}{2}\,\sqrt{T\,}$ and
$d_2(\sigma) = - \frac{\sigma}{2}\,\sqrt{T\,}=-d_1$,
hence~$\mathscr{C}(\sigma)$ reduces to
\[
\mathscr{C}(\sigma)
  =  S\,N\Bigl(\frac{\sigma}{2}\,\sqrt{T\,}\Bigr) - S\,N\Bigl(-\frac{\sigma}{2}\,\sqrt{T\,}\Bigr)
  =  \frac{S}{\sqrt{2\pi\,}} \int_{-\sigma\,\sqrt{T\,}/2}^{\sigma\,\sqrt{T\,}/2} e^{-t^2/2} \, dt
  =  S\,\erf\biggl(\sigma\,\sqrt{\frac{T}{8}\,}\biggr) \, ,
\]
where $\erf$ is~the \emph{error function} defined by
\[
\erf(z)
:=  \frac{2}{\sqrt{\pi\,}} \int_0^z e^{-t^2} \ dt \, .
\]

Also in this case, we consider an approximation of $\erf(z)$
which is based on the~hyperbolic tangent~\eqref{E-tanh}.
For example, following Ingber~\cite{Ingber1982},
a good approximation of~$\erf(z)$ may be
\[
\Theta_0(z)
  :=  \tanh\Bigl(\frac{2}{\sqrt{\pi\,}}\,z\Bigr)
\]
which has the~same limit at infinity and the~same derivative in~$0$.

A better approximation of~$\erf(z)$ is~the function
\[
\Theta_1(z)
  :=  \tanh \biggl( \frac{2}{\sqrt{\pi}} \, z + \frac{8-2\pi}{3\sqrt{\pi^3}} \, z^3 \, \biggr) \, ,
\]
which has the~same Taylor expansion of~order~3 in~0 of~the error function,
or the~function
\[ 
\Theta_2(z)
   :=  \tanh(a\, z +  b \, z^3) \, ,
\qquad
\text{with $a=1.129324$ and $b = 0.100303$} \, ,
\]
obtained in~\cite{Fairclough2000} by an optimization procedure. 

Numerical tests show that $\Theta_1(z)$ and $\Theta_2(z)$ provide approximations of the same order of accuracy, 
(see~\textsection\ref{SS:erf-vs-Theta} for details).

Hence we~have that $\mathscr{C}(\sigma)$ can be approximated by
\begin{align} \label{E-12}
\wh{\mathscr{C}}_0(\sigma)
&  =  S\,\tanh\biggl(\sigma\,\sqrt{\frac{T}{2\,\pi}\,}\biggr)
\intertext{%
or}
\wh{\mathscr{C}}_1(\sigma)
&  =  S\,\tanh\biggl(\, \Bigl( \sigma\,\sqrt{\frac {T}{2\pi}\,} \Bigr)
                      +  \frac{4-\pi}{12} \, \Bigl(\sigma\,\sqrt{\frac{T}{2\pi}\,}\Bigr)^3\biggr) \, , \label{E-13}
\intertext{%
or}
\wh{\mathscr{C}}_2(\sigma)
&  =  S\,\tanh\biggl( a\,\Bigl( \sigma\,\sqrt{\frac{T}{8}\,} \Bigr)
                      +  b\,\Bigl( \sigma\,\sqrt{\frac{T}{8}\,} \Bigr)^3\biggr) \, , \label{E-14}
\end{align}
with $a = 1.129324$ and $b = 0.100303$.

\section{Approximation of the implied volatility}

In order to find the implied volatility we should invert the call function $\mathscr{C}$,
i.e.\ we should be able to solve the equation $\mathscr{C}(\sigma)=C$.
Since $\wh{\mathscr{C}}$ is a good approximation of $\mathscr{C}$,
we~approximate the implied volatility by solving the equation
\begin{equation} \label{E-ChatC}
\wh{\mathscr{C}}(\sigma) = C \, .
\end{equation}

\subsection{Case $S \ne X$}

In this case, by \eqref{E-Chat}, equation \eqref{E-ChatC} is equivalent to
\[
\wh{\chi}_\alpha \Bigl(\frac{\sigma\,\sqrt{T\,}}{\alpha}\Bigr)
  =  C^* \, ,
\qquad
\text{where}
\quad
C^*
  :=  \begin{cases}
      C/S  &  \text{if $X>S$} \, ,  \\*[2mm]
      (C-S+X)/X  &  \text{if $X<S$} \, .
      \end{cases}
\]

According to~\eqref{E-fcirc}
such an equation is equivalent to
\begin{equation} \label{E-philog}
\varphi_\alpha\Bigl(\frac{\sigma\,\sqrt{T\,}}{\alpha}\Bigr)
  =  \frac{1}{2} \, \log\Bigl(\frac{C^*}{1-C^*}\Bigr) 
  =  \frac{1}{2} \, \log\Bigl(\frac{C - [S-X]^+}{S-C}\Bigr) \, ,
\end{equation}
where $[S-X]^+ = \max(S-X,0)$ is the pay-off. 

On the other hand, for any $\lambda\in\R$ the~equation 
\[
c_1\,x - \frac {c_2} x  + c_3 = \lambda \,
\qquad \text{i.e.} \qquad
c_1\,x^2 - (\lambda-c_3)\,x-c_2 = 0
\]
has a unique positive solution, given by
\[
x
  =  \frac{1}{2\,c_1} \Bigl[\lambda - c_3 + \sqrt{(\lambda-c_3)^2 + 4\,c_1\,c_2 \, }\Bigr] \, .
\]

Thus the implied volatility can be approximated by
\begin{equation}\label{E-inv}
\widehat{\sigma}
  =  \frac{\alpha}{2\,c_1(\alpha) \, \sqrt{T\,}}
     \left[\Lambda - c_3(\alpha)
            + \sqrt{\bigl(\Lambda - c_3(\alpha)\bigr)^2
            	      + 4\,c_1(\alpha) \,c_2(\alpha) \, }\right] \, ,
\end{equation}
where $\Lambda= \frac{1}{2} \, \log\Bigl(\frac{C - [S-X]^+}{S-C}\Bigr)$ and
$c_1(\alpha)$, $c_2(\alpha)$ and  $c_3(\alpha)$ are given by \eqref{E-c1}, \eqref{E-c2} and~\eqref{E-c3}.

\subsection{Case $S=X$}

In this case a first approximation of $\mathscr{C}$
is given by $\wh{\mathscr{C}}_0$ defined by \eqref{E-12} and therefore the approximating
equation $\wh{\mathscr{C}}_0(\sigma)=C$
has the solution
\begin{equation} \label{E-MOT11}
\wh{\sigma}_0
  =  \sqrt{\frac{\pi}{2\,T}\,} \, \log\Bigl(\frac{S+C}{S-C}\Bigr) \, .
\end{equation}

Better approximations of $\mathscr{C}$ are given by $\wh{\mathscr{C}}_1$ and $\wh{\mathscr{C}}_2$
defined by \eqref{E-13} and \eqref{E-14}.
Then the  equations $\wh{\mathscr{C}}_1(\sigma) = C$ and $\wh{\mathscr{C}}_2(\sigma) = C$
are equivalent to the equations 
\begin{equation} \label{E-131}
\biggl(\sigma\,\sqrt{\frac{T}{2\pi}\,}\biggr)
+  \frac {4-\pi}{12}\,\biggl(\sigma\,\sqrt{\frac{T}{2\pi}\,}\biggr)^3
=  \frac{1}{2} \, \log\Bigl(\frac{S+C}{S-C}\Bigr) \, .
\end{equation}
and 
\begin{equation} \label{E-141}
a\,\biggl(\sigma\,\sqrt{\frac{T}{8}\,}\biggr)
+  b\,\biggl(\sigma\,\sqrt{\frac{T}{8}\,}\biggr)^3
=  \frac{1}{2} \, \log\Bigl(\frac{S+C}{S-C}\Bigr) \, .
\end{equation}
respectively.

Now it is well known that for all $p>0$ the equation
\[
x^3 + 3\,p\,x = 2\,q
\]
has a unique \emph{real} solution given by
\[
x
  =  \sqrt[3]{  \sqrt{ p^3 + q^2 \,} + q } \;
     - \; \sqrt[3]{ \sqrt{ p^3 + q^2 \,} - q \,} \, .
\]
Hence the unique solution to equation~\eqref{E-131} is~given by
\begin{equation} \label{E-SigmaMOT2}
\wh{\sigma}_1
  :=  \sqrt{\frac{2\pi}{T}\,} \biggl[\sqrt[3]{  \sqrt{ p^3 + q^2 \,} + q } \;
     - \; \sqrt[3]{ \sqrt{ p^3 + q^2 \,} - q \,}\biggr] \, , 
\end{equation}
with
\[
p := \frac{4}{4-\pi}
\qquad \text{and} \qquad  
q := \frac{3}{4-\pi} \, \log\Bigl(\frac{S+C}{S-C}\Bigr) \, .
\]

Similarly the unique solution of equation~\eqref{E-141} is~given by
\begin{equation} \label{E-SigmaMOT3}
\wh{\sigma}_2
  :=  \sqrt{\frac{8}{T}\,} \biggl[\sqrt[3]{ \sqrt{ p^3 + q^2 \,} \, + q } \;
     - \; \sqrt[3]{ \sqrt{ p^3 + q^2 \,} \, - \, q } \; \biggr] \, ,
\end{equation}
where
\[
p := \frac{a}{3\,b}
\qquad \text{and} \qquad 
q := \frac{1}{4b} \, \log\Bigl(\frac{S+C}{S-C}\Bigr) \, .
\]

Thus we can conclude that
$\wh{\sigma}_0$, $\wh{\sigma}_1$ and $\wh{\sigma}_2$
given by \eqref{E-MOT11}, \eqref{E-SigmaMOT2} and \eqref{E-SigmaMOT3}
are approximating values of the implied volatility.

\section{Numerical results}

\subsection{Call when $S \ne X$}

By a numerical inspection we can see that the approximating function~$\wh{\chi}_\alpha(x)$ replicates well the standardized call function~$\chi_\alpha(x)$, except from the part on the far right hand side of the inflection point where the curvature is higher (and the volatility unrealistic). 

Therefore, to see better where the differences are and their magnitude we~studied the comparisons between $\chi_\alpha(x)$ and the~approximating function~$\wh{\chi}_\alpha(x)$ in~\eqref{E-Chatalpha} for one of the most common maturity (i.e.~$T=0.25$).

\subsubsection{Monte Carlo analysis}

Given $T=0.25$, we considered the difference $\chi_\alpha(x)-\wh{\chi}_\alpha(x)$ for 10,000 moneyness and 500 instances of $\sigma $ uniformly distributed between $[0\,,1.25]$. We first take the whole interval~$[0\,,1.25]$ and then we~divide it into five parts (each containing 100 $\sigma $) to~illustrate where the~differences are higher or lower.

We noticed that the biggest errors are when $\sigma \in [0.75\,,1.25]$
i.e.~for those cases in which the volatility is extremely rare (see Fig.~\ref{fig:VIX}).

\subsection{Call when $S = X$} \label{SS:erf-vs-Theta}

Here we start by plotting
the graph of the error function and~the approximating functions $\Theta_0$, $\Theta_1$ and~$\Theta_2$
as defined in Section \ref{S:S=X} then
we compare numerically some significant values.

To show the approximation's error
we plot the graph of the differences between the $\erf(z)$
and the functions $\Theta_0, \Theta_1$ and $\Theta_2$ as defined in Section \ref{S:S=X},
using a different scale for the $y$-axis.

Next we~compute some statistical values of  the  differences
between the error function and the functions $\Theta_0, \Theta_1$ and $\Theta_2$.

\subsubsection{Monte Carlo analysis}

Now we consider the lattice $\mathfrak{L}$ composed of the couples $(\sigma, T)$ with $\sigma= 10^{-4} \,j$, with $j=1,\dotsc, 10^{4}$ and $T =k/12 $ with $k=1,\dotsc, 24$. From it we randomly extract 10,000 samples in each interval to derive $z$. The statistics on the errors confirm the good quality of the approximation.

\subsection{Implied volatility}

As already mentioned there are several methods available in~literature
for deriving the~implied volatility through an approximated formula.
In Orlando and Taglialatela (2017)~\cite{Orlando2017}
we~compared the~results derived with Brenner \&~Subrahmanyam~\cite{Brenner1998},
Corrado \&~Miller~\cite{Corrado1996a},~\cite{Corrado1996} and Li~\cite{Li2005} formulae.
As we~found that the~latter is~the most accurate, we~consider Li formula as our benchmark.

\subsubsection{Case when $S \ne X$}

We compare the results obtained by Li formula denoted by $\widehat{\sigma}_{L}$
with those obtained with formula \eqref{E-inv} denoted by $\widehat{\sigma}$.
The prices of~all calls have been generated with the~BS model and, then,
the implied volatility has been derived by using the inversion formulae.
Each column provides the results of the said formulae for maturities $T$ from $0.1$ to $1.5$ versus the true volatility.

As shown, $\widehat{\sigma}$ is available even when $\widehat{\sigma}_{L}$ is not
and approximates better the~implied volatility for all maturities,
moneyness and level of~$ \sigma $ except the~part where it is~very high
(even though this occurrence is~quite unlikely as mentioned in Section \ref{SS:LitReview}).
Moreover the~error for the~$\widehat{\sigma}$ formula is~more consistent.

\subsubsection{Case when $S = X$}

Let us now consider the ATM case.
We compare the implied volatility $\widehat{\sigma}_{L}$ calculated with Li formula~\cite{Li2008}, $\widehat{\sigma}_1$ and $\widehat{\sigma}_2$ as defined in \eqref{E-SigmaMOT2} and~\eqref{E-SigmaMOT3}.
The prices of~all calls are generated with the~BS model for a given volatility ranging from 15\% to~125\% and maturity comprised between 0.1 and~1.5 years.
We also perform some statistics on the error between the~true volatility and the~one estimated.

\section{Conclusions}

In this work we~have recalled the~importance of~calculating the~value of~the call for pricing as well as for inferring the~implied volatility. Even though calculations can be easily performed by using numerical methods we share the Li opinion~\cite{Li2005} that ``to simplify some applications such as spreadsheet, it~may be useful to~have an approximation formula if that formula is~reasonably simple, accurate and valid for a wide range of~cases. The cost and inconvenience of~iterating also motivate the~search for explicit formulas''. 

Here a standardized call function has been introduced to~represent the~whole family of~calls and to simplify the calculations.
In addition we~have shown how the~approximation of the aforementioned standardized call can be performed through hyperbolic tangents instead of~the usual Taylor truncation. This allows a greater accuracy for extreme values of $\sigma$, which makes this approach particularly suitable for stress testing and hedging purposes. Finally we~have derived some explicit formulae for approximating the~implied volatility that seem to be~superior to~the ones proposed in~literature so far and are~valid regardless of~option's moneyness.
Therefore because of~ the higher accuracy and flexibility, this approach could replace current methods with little additional effort.
Further research will address the cases when there is a marked difference between the approximation~$\wh{\mathscr{C}}(\sigma)$
and the~BS call~$\mathscr{C}(\sigma)$ in order to provide a better approximation.

\end{document}